\documentclass{amsart}

\usepackage{amsmath}
\usepackage{amssymb}
\usepackage{epsfig}  		% For postscript
\usepackage{xypic}
\usepackage{bbm}
\usepackage{hyperref}
\usepackage{eucal}
\usepackage{mathrsfs}
\usepackage{cancel}
\usepackage{stmaryrd}
\usepackage{xcolor}
\usepackage{lineno}

\newtheorem{mthm}{Theorem}

\newtheorem{thm}{Theorem}[section]
\newtheorem{prop}[thm]{Proposition}
\newtheorem{lem}[thm]{Lemma}
\newtheorem{cor}[thm]{Corollary}

\theoremstyle{definition}

\theoremstyle{remark}
\newtheorem{remark}[thm]{Remark}

\numberwithin{equation}{section}

\DeclareMathAlphabet{\mathpzc}{OT1}{pzc}{m}{it}
\renewcommand{\mathcal}[1]{\mathpzc{#1}}

\newcommand{\RR}{\mathbbm{R}}

\newcommand{\C}{\mathrm{C}}
\newcommand{\T}{\mathrm{T}}

\newcommand{\nablae}{\nabla^{({\rm e})}}
\newcommand{\nablam}{\nabla^{({\rm m})}}
\newcommand{\nablak}{\nabla^{\kappa}}

\renewcommand{\bar}{\overline}

\DeclareMathOperator{\dist}{\mathrm{dist}}

\newcommand{\group}{\mathrm}

\newcommand{\Iso}{\group{Iso}}
\newcommand{\Aff}{\group{Aff}}
\newcommand{\GL}{\group{GL}}

\newcommand{\SL}{\group{SL}}

\newcommand{\SO}{\group{SO}}

\renewcommand{\O}{\group{O}}

\newcommand{\tr}{\mathrm{tr}}

\renewcommand{\d}{\mathrm{d}}

\DeclareMathOperator{\diag}{\mathrm{diag}}

\newcommand{\domain}{\mathcal}

\newcommand{\doN}{\domain{N}}

\newcommand{\Sym}{\mathrm{Sym}}

\newcommand{\Pos}{\mathrm{Pos}}
\newcommand{\Diag}{\mathrm{Diag}}

\newcommand{\frg}{\mathfrak{g}}

\renewcommand{\frm}{\mathfrak{m}}

\newcommand{\frk}{\mathfrak{k}}

\renewcommand{\sl}{\mathfrak{sl}}

\newcommand{\tensor}{\mathsf}

\newcommand{\g}{\tensor{g}}
\newcommand{\R}{\tensor{R}}

\renewcommand{\rho}{\varrho}
\renewcommand{\phi}{\varphi}
\renewcommand{\d}{\mathrm{d}}

\begin{document}

\baselineskip=0.48cm

%\linenumbers

\title[Information geometry on normal distributions]{Information geometry
and asymptotic geodesics on the space of normal distributions}

\author[Globke]{Wolfgang Globke}

\address{Wolfgang Globke,
Faculty of Mathematics\\
University of Vienna\\
Oskar-Morgenstern-Platz 1\\
1090 Vienna\\
Austria}
\email{wolfgang.globke@univie.ac.at}

\author[Quiroga-Barranco]{Ra\'ul Quiroga-Barranco}
\address{Ra\'ul Quiroga-Barranco,
Centro de Investigaci\'on en Matem\'aticas,
A.P. 402,
Guanajuato, Gto. C.P. 36000,
Mexico}
\email{quiroga@cimat.mx}

\thanks{Wolfgang Globke was partially supported by the Australian Research Council grant {DE150101647} and the Austrian Science Fund FWF grant I 3248.
Ra\'ul Quiroga-Barranco was partially supported by a CONACYT grant.}


\begin{abstract}
The family $\doN$ of $n$-variate normal distributions is
parameterized by the cone of positive definite symmetric
$n\times n$-matrices and the $n$-dimensional real vector space.
Equipped with the Fisher information metric, $\doN$ becomes a Riemannian
manifold. As such, it is diffeomorphic, but not isometric, to the
Riemannian symmetric space $\Pos_1(n+1,\RR)$ of unimodular positive definite
symmetric $(n+1)\times(n+1)$-matrices.
As the computation of distances in the Fisher metric for $n>1$ presents some
difficulties, Lovri\v{c} et al.~(2000) proposed to use the Killing metric
on $\Pos_1(n+1,\RR)$ as an alternative metric in which distances are easier
to compute.
In this work, we survey the geometric properties of the space $\doN$
and provide a quantitative analysis of the defect of certain geodesics for
the Killing metric to be geodesics for the Fisher metric. We find that
for these geodesics
the use of the Killing metric as an approximation for the Fisher metric is
indeed justified for long distances.
\end{abstract}

\setcounter{tocdepth}{1}

\maketitle

%%%%%%%%%
\section{Introduction and overview}
\label{sec:intro}
%%%%%%%%%

A multivariate normal distribution is determined by its
covariance matrix and its mean vector.
So for a fixed $n\geq 1$, the family $\doN$ of $n$-variate normal
distributions is a differentiable manifold which can be identified with
the product of the space of positive definite symmetric
$n\times n$-matrices by the vector space $\RR^n$.
For various statistical purposes, it is desirable to have a measure of
distance between the elements of $\doN$.
Such a distance measure is provided by the \emph{Fisher metric} on $\doN$,
which is a Riemannian metric that appears naturally in a certain statistical
framework.
We briefly review some properties of Fisher metric on the normal
distributions in Section \ref{sec:background}.

Computing the distances on $\doN$, however, turns out to be a non-trivial
task. Even though explicit forms for the geodesics of the Fisher metric on
$\doN$ are known (due to Calvo and Oller \cite{CO2}), these only yield
explicit formulas for the distance in particular cases.
So Lovri\v{c}, Min-Oo and Ruh \cite{LMR} proposed the use of a different
metric in which distances are easier to compute.
They map $\doN$ diffeomorphically onto the Riemannian symmetric space
$\SL(n+1,\RR)/\SO(n+1)$.
This map is not an isometry between the Fisher metric and the metric of the
symmetric space, which we call the \emph{Killing metric}, but
nevertheless, the two metrics are quite similar in appearance.
So it is reasonable to ask how different they really are.

In Section \ref{sec:normal} we describe the geometry of $\doN$ as a
Riemannian homogeneous but non-symmetric space with the Fisher metric.
In Theorem \ref{mthm:normal_fisher} we show that $\doN$ is a bundle whose
base is the cone $\Pos(n,\RR)$ of symmetric positive definite
$n\times n$-matrices and whose fiber is $\RR^n$.
This also gives rise to two pointwise mutually orthogonal foliations, one
with leaves isometric to $\Pos(n,\RR)$, the other with leaves isometric
to $\RR^n$.

To make a case for using the Killing metric as a sensible approximation
for the Fisher metric, we compare the geometry of the Fisher metric and the
geometry of the Killing metric in Section \ref{sec:symmetric}.
We find that the Levi-Civita connection for the Fisher metric on the leaves
$\Pos(n,\RR)$ is affinely equivalent to the Levi-Civita connection of the
Killing metric.
So unparameterized geodesics in these leaves are the same for the two
metrics.
In Theorem \ref{mthm:asymptotic}, we show that Killing geodesics orthogonal
to a leaf $\Pos(n,\RR)$ at some point are \emph{asymptotically geodesic} in
the Fisher metric, that is, their defect from being a Fisher geodesic
tends to zero as their curve parameter tends to infinity.
So we find that for two important classes of unparameterized geodesics,
the Killing geodesics approximate or are identical to the corresponding
Fisher geodesics.
Though this is not an exhaustive comparison, it provides some justification
to consider the easier to compute Killing metric as a good approximation
for the Fisher metric.

%%%
\subsection*{Notations and conventions}

Throughout, we will assume matrices to be real-valued.
For a matrix $X\in\RR^{n\times n}$, we let $X^\top$ denote its
transpose. We also write $X^{-\top}=(X^\top)^{-1}$.
The identity matrix is denoted by $I$ or $I_n$.
By $E_{ij}$ we denote the elementary matrix whose entry in row $i$, column
$j$ is $1$, and all other entries are $0$. Its symmetrization is
$S_{ij}=\frac{2-\delta_{ij}}{2}(E_{ij}+E_{ji})$.
%\RQ{I replaced the expression for $S_{ij}$, which I believe is now correct. The original expression was $S_{ij}=\frac{2}{2-\delta_{ij}}(E_{ij}+E_{ji})$}
The canonical basis vectors of $\RR^n$ are denoted by $e_1,\ldots,e_n$.

As usual,
\begin{align*}
\GL(n,\RR) &= \{A\in\RR^{n\times n}\mid \det(A)\neq0\}, \\
\SL(n,\RR) &= \{A\in\GL(n,\RR)\mid\det(A)=1\}, \\
\O(n) &= \{A\in\GL(n,\RR)\mid A^\top = A^{-1}\}, \\
\SO(n) &= \{A\in\O(n)\mid \det(A)=1\}
\end{align*}
denote the general linear, special linear, and (special) orthogonal groups,
respectively.
The subgroup of $\GL(n,\RR)$ of matrices with positive
determinant is denoted by $\GL^+(n,\RR)$.
The affine group is the semidirect product
\[
\Aff(n,\RR)=\GL(n,\RR)\ltimes\RR^n,
\]
where the semidirect product is given by
$(A_1,b_1)(A_2,b_2)=(A_1 A_2,b_1+A_1 b_2)$ for $(A_i,b_i)\in\Aff(n,\RR)$.
We also write $\Aff^+(n,\RR)=\GL^+(n,\RR)\ltimes\RR^n$.

By $\Sym(n,\RR)$ we denote the set of symmetric $n\times n$-matrices,
\[
\Sym(n,\RR) = \{S\in\RR^{n\times n}\mid S=S^\top\}.
\]
We write $\Sym_0(n,\RR)$ for the corresponding subspaces of elements with
trace $0$.
The subset of diagonal matrices in $\Sym(n,\RR)$ is denoted by
$\Diag(n,\RR)$.

The set of positive definite symmetric matrices in $\Sym(n,\RR)$
is denoted by $\Pos(n,\RR)$,
\[
\Pos(n,\RR) = \{S\in\Sym(n,\RR)\mid x^\top S x > 0 \text{ for all non-zero } x\in\RR^n\}.
\]
Its subset of unimodular elements is
\[
\Pos_1(n,\RR) = \{S\in\Pos(n,\RR)\mid \det(S)=1\}.
\]
Recall that $\Pos(n,\RR)=\GL(n,\RR)/\O(n)$ and
$\Pos_1(n,\RR)=\SL(n,\RR)/\SO(n)$.

%%%%%%%%%
\section{Some background on information geometry}
\label{sec:background}
%%%%%%%%%

In this section we briefly review the concepts from information
geometry that we use in the following.
We mainly follow Amari and Nagaoka's \cite{AN} presentation.

%%%
\subsection{The Fisher metric and dual connections}

Information geometry provides a framework to study a class
of probability distributions $p(x;\theta)$ defined
on a sample space $\Omega$ and determined by finitely
many parameters $\theta=(\theta_1,\ldots,\theta_n)$,
where we assume for simplicity that $p$ depends smoothly on $x$
and $\theta$.
For example, the set of univariate normal distributions is
parametrized by the mean $\theta_1=\mu$ and the variance
$\theta_2=\sigma^2$.

In general, the set $M$ of admissible values for $\theta$ can be
viewed as an $n$-dimensional differentiable manifold, and we can
define a positive semidefinite bilinear tensor $\g=(\g_{ij})$
on $M$ via
\begin{equation}
\g_{ij}(\theta) = -\int_\Omega \frac{\partial^2 \log(p(x;\theta))}{\partial\theta_i\partial\theta_j} p(x;\theta)\ \d x.
\label{eq:fisher}
\end{equation}
In the following we assume that $\g$ is positive
definite everywhere, so that $(M,\g)$ is a Riemannian manifold.
Then $\g$ is called the \emph{Fisher metric} on $M$, and $(M,\g)$
is called a \emph{statistical manifold}.

In addition to the Fisher metric, there are two particular
torsion-free affine connections defined on $M$, denoted by
$\nablae$ and $\nablam$.
These connections are \emph{dual} to each other with respect to
$\g$, which means that for all vector fields $X,Y,Z$
on $M$,
\begin{equation}
Z \g(X,Y) = \g(\nablae_Z X,Y) + \g(X,\nablam_Z Y).
\label{eq:dual}
\end{equation}
Moreover, the affine combination
\[
\nabla = \frac{1}{2}\nablae + \frac{1}{2}\nablam
\]
yields the Levi-Civita connection $\nabla$ of the Fisher metric $\g$.

The letters ``e'' and ``m'' stand for ``exponential'' and
``mixture'', respectively, referring to two families of probability
distributions in which these connections appear naturally.
More generally, there is a whole family of affine connections
$\nabla^{(\alpha)}$ with $\alpha\in[-1,1]$ associated to $\g$,
and $\nablae=\nabla^{(1)}$, $\nablam=\nabla^{(-1)}$. However,
we are not concerned with values $\alpha\neq\pm 1$ here.

%%%
\subsection{Exponential families}

An \emph{exponential family} is a statistical manifold $M$ that
consists of probability distributions of the form
\[
p(x;\theta) = \exp(c(x)+\theta_1 f_1(x) + \ldots + \theta_n f_n(x) - \psi(\theta))
\]
for given functions $c,f_1,\ldots,f_n:\Omega\to\RR$ and
$\psi:M\to\RR$.
The normalization of $p(x;\theta)$ implies
\begin{equation}
\psi(\theta) = \log\left(\int_\Omega \exp(c(x) + \theta_1 f_1(x) + \dots + \theta_n f_n(x)) \d x\right).
\label{eq:psi}
\end{equation}

The connections $\nablae$ and $\nablam$ are distinguished on
an exponential family (see Amari and Nagaoka \cite[Sections 2.3 and 3.3]{AN}).

\begin{thm}\label{thm:flat}
Let $M$ be an exponential family. Then $\nablae$ and $\nablam$
are flat torsion-free affine connections on $M$.
\end{thm}

In fact, the $\theta_1,\ldots,\theta_n$ form a
\emph{flat coordinate system}
in the sense that $\nablae_{\partial_i}\partial_j=0$, $i,j=1,\ldots,n$,
for the coordinate vector fields $\partial_i=\frac{\partial}{\partial\theta_i}$.
The flat coordinate system $\eta_1,\ldots,\eta_n$ for
$\nablam$ is obtained via a Legendre transform of
$\theta_1,\ldots,\theta_n$,
\[
\frac{\partial \psi}{\partial \theta_i} = \eta_i,
\quad i=1,\ldots,n.
\]
In the flat $\theta$-coordinates, the Fisher metric for an
exponential family is given as a Hessian metric
$\g=\nablae\d\psi$, or equivalently
\begin{equation}
\g_{ij}(\theta) = \frac{\partial^2 \psi(\theta)}{\partial\theta_i\partial\theta_j}.
\label{eq:hessian}
\end{equation}
We call $\psi$ the \emph{potential} of the Fisher metric.
The \emph{dual potential} $\psi^*$ is given by
$\psi^*=\theta^\top\eta-\psi$, and in the flat
$\eta$-coordinates, the inverse $\g^{ij}$ is given as a
Hessian metric
\begin{equation}
\g^{ij}(\eta) = \frac{\partial^2 \psi^*(\eta)}{\partial\eta_i\partial\eta_j}.
\label{eq:hessian_dual}
\end{equation}

Another important property of exponential families is the
following (see Amari and Nagaoka \cite[Theorem 2.5]{AN}).

\begin{thm}\label{thm:totgeod}
A submanifold $N$ of an exponential family $M$ is totally
geodesic in $M$ with respect to $\nablae$ if and only if $N$ is
an exponential family itself.
\end{thm}

%%%
\subsection{Normal distributions}
\label{subsec:normal}

The most important exponential family is formed by the normal
distributions.
An $n$-variate normal distribution is determined by its covariance matrix $\Sigma\in\Pos(n,\RR)$ and its mean $\mu\in\RR^n$ by the following formula
\[
p(x; \Sigma, \mu) = \frac{1}{\sqrt{(2\pi)^n\det(\Sigma)}}
    \exp\left(-\frac{1}{2}(x-\mu)^\top \Sigma^{-1} (x-\mu)\right)
\]
so the manifold we are considering is the space
$\doN=\Pos(n,\RR)\times\RR^n$.
The flat coordinates for the connection $\nablam$ are
$(\Xi,\xi)$, where
\[
\xi=\mu\in\RR^n,
\quad
\Xi=\Sigma+\mu\mu^\top\in\Pos(n,\RR),
\]
and the flat coordinates for the connection $\nablae$ are $(\Theta,\theta)$,
where
\[
\theta=\Sigma^{-1}\mu\in\RR^n,
\quad
\Theta=-\frac{1}{2}\Sigma^{-1}\in\Pos(n,\RR).
\]
The potential $\psi$ in these coordinate systems is
(compare \eqref{eq:psi})
\begin{align*}
\psi(\Sigma,\mu)
&=
\frac{1}{2}\mu^\top \Sigma^{-1} \mu
+
\frac{1}{2}\log(\det(2\pi\Sigma)), \\
\psi(\Xi,\xi)
&=
\frac{1}{2}\xi^\top(\Sigma-\xi\xi^\top)^{-1}\xi
+\frac{1}{2}\log(\det(2\pi(\Xi-\xi\xi^\top))), \\
\psi(\Theta,\theta) &= -\frac{1}{4}\theta^\top\Theta\theta-\frac{1}{2}\log(\det(-\pi^{-1}\Theta)).
\end{align*}

%%%%%%%%%
\section{Geometry of the family of normal distributions}
\label{sec:normal}
%%%%%%%%%

In this section we take a closer look at the information
geometry of the manifold $\doN=\Pos(n,\RR)\times\RR^n$.
Note that $\Pos(n,\RR)=\RR\times\Pos_1(n,\RR)$ as a
product of manifolds.

%%%
\subsection{Basic geometric properties of $\doN$}

Here, we state the explicit form of the Fisher metric, its Levi-Civita
connection and its curvature tensor in the $(\Sigma,\mu)$-coordinates.
These were originally computed by Skovgaard \cite{skovgaard0,skovgaard}.

If $\g$ is the Fisher metric on $\doN$, $X,Y$ are two coordinate
vector fields in the $\Sigma$-directions, and $v,w$ are two
coordinate vector fields in the $\mu$-directions, then the metric tensor is
\begin{equation}
\g_{(\Sigma,\mu)}\bigl( (X,v),(Y,w) \bigr)
=v^\top\Sigma^{-1} w + \frac{1}{2}\tr(\Sigma^{-1}X\Sigma^{-1}Y),
\label{eq:Fisher_metric}
\end{equation}
and the Levi-Civita connection is determined by
\begin{equation}
\begin{split}
\nabla_XY &= \nabla_YX = -\frac{1}{2}(X\Sigma^{-1}Y+Y\Sigma^{-1}X), \\
\nabla_vw &= \nabla_wv = \frac{1}{2}(vw^\top+wv^\top), \\
\nabla_Xv &= \nabla_vX = -\frac{1}{2}X\Sigma^{-1}v.
\end{split}
\label{eq:LeviCivita}
\end{equation}
Note that the symmetry in these equations is due to the fact that we are
looking at coordinate vector fields.

If $X_1,X_2,X_3,X_4$ and $v_1,v_2,v_3,v_4$ are coordinate vector fields
in the $\Sigma$- and $\mu$-directions, respectively, then the curvature
of the Fisher metric is determined by
\begin{equation}
\begin{split}
\R(v_1,v_2,v_3,v_4) =&\ \frac{1}{4}\bigl((v_2^\top\Sigma^{-1}v_3)(v_1^\top\Sigma^{-1}v_4)-(v_1^\top\Sigma^{-1}v_3)(v_2^\top\Sigma^{-1}v_4)\bigr), \\
\R(X_1,X_2,X_3,X_4) =&\ \frac{1}{4}\bigl(\tr(X_2\Sigma^{-1}X_1\Sigma^{-1}X_3\Sigma^{-1}X_4\Sigma^{-1})\\
&\ -\tr(X_1\Sigma^{-1}X_2\Sigma^{-1}X_3\Sigma^{-1}X_4\Sigma^{-1})\bigr), \\
\R(v_1,v_2,X_1,X_2) =&\ \frac{1}{4}(v_1^\top\Sigma^{-1}X_1\Sigma^{-1}X_2\Sigma^{-1}v_2-v_1^\top\Sigma^{-1}X_2\Sigma^{-1}X_1\Sigma^{-1}v_2), \\
\R(v_1,X_1,v_2,X_2) =&\ \frac{1}{4}v_1^\top\Sigma^{-1}X_1\Sigma^{-1}X_2\Sigma^{-1}v_2.
\end{split}
\label{eq:curvature}
\end{equation}

We now consider the two foliations of $\doN$ into submanifolds of fixed
$\Sigma_0$ or $\mu_0$, respectively.
For fixed $\Sigma_0\in\Pos(n,\RR)$, $\mu_0\in\RR^n$ we will
write
\begin{align*}
\doN(\cdot,\mu_0) &= \{(\Sigma,\mu_0)\mid\Sigma\in\Pos(n,\RR)\}, \\
\doN(\Sigma_0,\cdot) &= \{(\Sigma_0,\mu)\mid\mu\in\RR^n\}.
\end{align*}
It follows from \eqref{eq:Fisher_metric} that the two foliations determined
by these submanifolds are orthogonal.

Recall that the second fundamental form $B$ of a submanifold $N$ of $M$ is
the normal component of $\nabla_X Y$ in $\T M$ for two vector fields
$X,Y$ tangent to $N$.
We let $\partial_{(ij)}$ denote the coordinate vector field in direction
$S_{ij}$, and we let $\partial_m$ denote the coordinate vector field in
direction $e_m$.
We denote by $J_\Sigma=\{(i,j)\mid 1\leq i\leq j\leq n\}$ the set
enumerating the coordinates of $\Sym(n,\RR)$ and by
$J_\mu=\{i=1,\ldots,n\}$ the set enumerating the coordinates of $\RR^n$,
and set $J=J_\Sigma\cup J_\mu$. When we refer to an index $p\in J$, it may
mean either a single index from $J_\mu$ or an index pair from $J_{\Sigma}$.
Then the Christoffel symbols for the Levi-Civita connection $\nabla$ are
denoted by $\Gamma_{pq}^r$ with $p,q,r\in J$.

\begin{prop}\label{prop:Nmu0_tot_geod}
    For any $\mu_0 \in \RR^n$ and with respect to the Fisher metric of $\doN$, the
    submanifold $\doN(\cdot,\mu_0)$ is totally geodesic.
\end{prop}
\begin{proof}
By \eqref{eq:LeviCivita}, $\nabla_{\partial_{(ij)}}\partial_{(kl)}$ is tangent
to $\doN(\cdot,\mu_0)$ for all $i,j,k,l$.
An arbitrary tangent vector field $X$ to
$\doN(\cdot,\mu_0)$ can be written as
%    By Lemma~\ref{lem:Gamma_Sigma_mu}(4) we have $\Gamma_{ij}^k = 0$
%    for all $i,j \in J_\Sigma, k \in J_\mu$.
%    Next, we observe that
%    the tangent vector fields to $\doN(\cdot,\mu_0)$ can be written
$X=\sum_{(i,j)\in J_\Sigma} w_{ij} \partial_{(ij)}$, with
$w_{ij}\in\C^\infty(\doN)$.
Then
        \begin{align*}
            \nabla_{\partial_{(ij)}}X
                &= \sum_{p\in J} (\partial_{(ij)} w_p + \sum_{q\in J}
                    \Gamma_{(ij)q}^p w_q) \partial_p \\
                &= \sum_{p\in J} (\partial_{(ij)} w_p + \sum_{(k,l)\in J_\Sigma}\Gamma_{(ij)(kl)}^p w_{kl}) \partial_p & (w_m=0\text{ for }m\in J_\mu) \\
                &= \sum_{(r,s)\in J_\Sigma} (\partial_{(ij)} w_{rs} + \sum_{(k,l)\in J_\Sigma}\Gamma_{(ij)(kl)}^{(rs)} w_{kl}) \partial_{rs} & (\Gamma_{(ij)(kl)}^m=0=w_m\text{ for }m\in J_\mu).
        \end{align*}
    This last expression is the induced covariant derivative on the
    submanifold $\doN(\cdot,\mu_0)$, since the $\mu$- and
    $\Sigma$-directions are orthogonal everywhere.
    Hence the second fundamental form of $\doN(\cdot,\mu_0)$ vanishes,
    which means $\doN(\cdot,\mu_0)$ is totally geodesic.
\end{proof}

\begin{prop}\label{prop:NSigma0_parallel}
    For any $\Sigma_0 \in \Pos(n,\RR)$ and with respect to the Fisher metric
    of $\doN$, the submanifold $\doN(\Sigma_0,\cdot)$ is parallel. Also, the
    second fundamental form $B$ of $\doN(\Sigma_0,\cdot)$ satisfies
    \[
        B(e_i,e_j) = \frac{1}{2}(E_{ij} + E_{ji})
    \]
    for all $i,j = 1, \dots, n$.
\end{prop}
\begin{proof}
    The second fundamental form of $\doN(\Sigma_0,\cdot)$ is given by
        \[
            B(\partial_i,\partial_j) = \sum_{(k,l)\in J_\Sigma} \Gamma_{ij}^{(kl)} \partial_{(kl)},
        \]
    where $i,j \in J_\mu$.

    Denote by $\nabla^\perp$ and $\bar{\nabla}$ the normal and induced
    connection for $\doN(\Sigma_0,\cdot)$, respectively.
    By \eqref{eq:LeviCivita}, $\bar{\nabla}$ is a flat connection on
    $\doN(\Sigma_0,\cdot)$.
    Then the covariant derivative of $B$ is given by ($i,j,m \in J_\mu$)
        \[
            (\nabla_{\partial_m} B)(\partial_i, \partial_j) =
            \nabla^\perp_{\partial_m}(B(\partial_i,\partial_j))
                - B(\bar{\nabla}_{\partial_m} \partial_i, \partial_j)
                - B(\partial_i, \bar{\nabla}_{\partial_m} \partial_j)
                = \nabla^\perp_{\partial_m}(B(\partial_i,\partial_j)),
        \]
    where the last identity holds since $\bar{\nabla}$ is flat and
    $\partial_i$ come from affine coordinates.
    Hence, we have for all $i,j,m \in J_\mu$
        \begin{align*}
            (\nabla_{\partial_m} B)(\partial_i, \partial_j) &=
            \nabla^\perp_{\partial_m}(B(\partial_i,\partial_j))  \\
            &= \nabla^\perp_{\partial_m}\Bigl(\sum_{(k,l)\in J_\Sigma} \Gamma_{ij}^{(kl)} \partial_{(kl)}\Bigr) \\
            &= \sum_{(k,l) \in J_\Sigma} (\partial_m \Gamma_{ij}^{(kl)}) \partial_{(kl)}
                + \sum_{(k,l),(r,s) \in J_\Sigma} \Gamma_{ij}^{(kl)} \Gamma_{(kl)m}^{(rs)} \partial_{(rs)}.
        \end{align*}
    In this expression, $\partial_m \Gamma_{ij}^{(kl)}=0$ and
    $\Gamma_{(kl)m}^{(rs)}=0$ due to equation in \eqref{eq:LeviCivita}.
    These computations imply that $\nabla^\perp B = 0$, in other words
    that $\doN(\Sigma_0,\cdot)$ is parallel.

    On the other hand, to compute $B$ we use \eqref{eq:LeviCivita},
    \begin{align*}
        B(e_i,e_j)% &= B(\partial_i, \partial_j)
        =\frac{1}{2}(e_i e_j^\top + e_j e_i^\top)
        =\frac{1}{2}(E_{ij}+E_{ji}),
%            &= \sum_{(k,l) \in J_\Sigma} \Gamma_{ij}^{(kl)} \partial_{(kl)} \\
%            &= \sum_{(k,l) \in J_\Sigma}
%                \frac{2}{2-\delta_{ij}} \delta_{ik}\delta_{jl} \partial_{(kl)}   \\
%            &= \frac{2}{2-\delta_{ij}} \partial_{(ij)} \\
%            &= \frac{2}{2-\delta_{ij}} S_{ij} = E_{ij} + E_{ji},
    \end{align*}
    where we have used the identification of basis vector with their corresponding
    partial differential operators.
\end{proof}

From the previous result the submanifold $\doN(\Sigma_0,\cdot)$ is not
totally geodesic. Hence, $\doN$ is not the Riemannian product of
$\doN(\cdot,\mu_0)$ and $\doN(\Sigma_0,\cdot)$ even though they are mutually
orthogonal.

%\RQ{It would be interesting to say something from the explicit expression
%of $B$.}

%%%
\subsection{$\boldsymbol{\doN}$ as a homogeneous space}

It is well-known that the affine group
$\Aff(n,\RR)$ acts transitively
on $\doN$ by
\begin{equation}
(A,b)\cdot(\Sigma,\mu) = (A\Sigma A^\top, A\mu+b),
\label{eq:action}
\end{equation}
where $A\in\GL(n,\RR)$, $b\in\RR^n$, $(\Sigma,\mu)\in\doN$.
Furthermore, the action remains transitive when restricted to $\Aff^+(n,\RR)$.
The tangent space $\T_{(I,0)}\doN$ can be identified with the
vector space $\Sym(n,\RR)\times\RR^n$.
Given $(\Sigma,\mu)\in\doN$ and $(X,v)\in\T_{(\Sigma,\mu)}\doN$,
the tangent action of $(A,b)\in\Aff(n,\RR)$ is
\begin{equation}
(A,b)\cdot(X,v) = (AXA^\top, Av).
\label{eq:action_tangent}
\end{equation}
Thus we can identify
%\RQ{Is it necessary to write $A\cdot\Sym(n,\RR)\cdot A^\top$? I think it is enough to write $\Sym(n,\RR)$.}\WG{well, we write $AXA^\top$ for the tangent
%vectors in the following, I thought writing it like this would make it
%more explicit where that comes from} \RQ{Agree. It makes the computations below more clear.}
\[
\T_{(\Sigma,\mu)}\doN
\ \cong\
(A\cdot\Sym(n,\RR)\cdot A^\top)\oplus\RR^n,
\]
where $AA^\top=\Sigma$.

\begin{lem}\label{lem:Aff_isometry}
The affine group $\Aff(n,\RR)$ acts transitively and isometrically on
$\doN$ by \eqref{eq:action}.
Moreover, if $R\subset\GL(n,\RR)$ denotes the subgroup of lower triangular
matrices with positive diagonal entries, then the subgroup $R\ltimes\RR^n$
acts simply transitively on $\doN$.
\end{lem}
\begin{proof}
The transitivity is a well-known fact.
It remains to check that \eqref{eq:action} is isometric.
The tangent action of $(A,b)\in\Aff(n,\RR)$ is
\eqref{eq:action_tangent}, hence
\begin{align*}
&\g_{(A\Sigma A^\top,A\mu+b)}( (AXA^\top,Av),(AXA^\top,Av) ) \\
&=
(Av)^\top (A\Sigma A^\top)^{-1} (Av)
+
\frac{1}{2}\tr((A\Sigma A^\top)^{-1} AXA^\top (A\Sigma A^\top)^{-1} AXA^\top) \\
&=
v^\top \Sigma^{-1} v
+
\frac{1}{2}\tr(A^{-\top}\Sigma^{-1} X \Sigma^{-1} X A^\top)
=
v^\top \Sigma^{-1} v
+
\frac{1}{2}\tr(\Sigma^{-1} X \Sigma^{-1} X) \\
&=\g_{(\Sigma,\mu)}((X,v),(X,v)).
\end{align*}
This shows that the action is isometric.

Note that $(A,b)\cdot(I,0)=(I,0)$ is equivalent to $A\in\O(n)$,
$b=0$. So the stabilizer of $\Aff(n,\RR)$ at $(I,0)$ is
$\O(n)$. From the Iwasawa decomposition $\GL(n,\RR)=\O(n) R$
it follows that $R\ltimes\RR^n$ acts simply transitively.
\end{proof}

%%%
\subsection[Geometry of $\Pos(n,\RR)$]{Geometry of $\boldsymbol{\Pos(n,\RR)}$}

As a consequence of Proposition \ref{prop:Nmu0_tot_geod} and Theorem
\ref{thm:totgeod}, the Fisher metric of the family $\doN(\cdot,\mu_0)$ of
normal distributions with mean $\mu_0$ coincides with the restriction of
the Fisher metric of $\doN$ to $\doN(\cdot,\mu_0)$.
Since all of these submanifolds are isometric, we may take $\mu_0=0$ for
convenience.
In the following, we will make explicit how $\doN(\cdot,0)$ with its
Fisher metric is isometric to a symmetric space
$\Pos(n,\RR)=\GL(n,\RR)/\O(n)$ with a suitably scaled Killing metric.

Consider the product of irreducible Riemannian symmetric spaces
\[
M = \RR\times\Pos_1(n,\RR),
\]
where its Riemannian metric $\g_M=\g_1\times\g_2$ is the product of
the metric $\g_1$, which is $\frac{1}{2}$ times the
multiplication on $\RR$, and the metric $\g_2$ on
$\Pos_1(n,\RR)$ given by $\g_{2,\Sigma}( X,Y ) = \frac{1}{2}\tr(\Sigma^{-1}X\Sigma^{-1}Y)$.
Let $\GL(n,\RR)$ act on $M$ via
\[
A\cdot(\alpha,\Sigma) = (\alpha+2\log(\det(A)),\ \det(A)^{-2}A\Sigma A^\top).
\]

\begin{lem}\label{lem:RxPos1_isometry}
The $\GL(n,\RR)$-action on $M$ given above is by isometries.
\end{lem}
\begin{proof}
The tangent action of $A\in\GL(n,\RR)$ at $(\alpha,\Sigma)$
on $(t,X)\in\T_{(\alpha,\Sigma)}M$ is
\[
\d A_{(\alpha,\Sigma)} (t,X)
=(t,\ \det(A)^{-2} A X A^\top).
\]
Hence
\begin{align*}
&\g_{M,A\cdot(\alpha,\Sigma)}(\d A_{(\alpha,\Sigma)}(t_1,X_1),\d A_{(\alpha,\Sigma)}(t_2,X_2)) \\
&=
\g_{M,A\cdot(\alpha,\Sigma)}((t_1,\det(A)^{-2}AX_1A^\top),(t_2,\det(A)^{-2}AX_2A^\top)) \\
&=\g_{1,\alpha+2\log(\det(A))}(t_1,t_2) + \g_{2,\det(A)^{-2}A\Sigma A^\top}(\det(A)^{-2}AX_1A^\top,\det(A)^{-2}AX_2A^\top) \\
&= \frac{1}{2}t_1 t_2 + \frac{1}{2}\tr((\det(A)^{-2}A\Sigma A^\top)^{-1}\det(A)^{-2}AX_1A^\top (\det(A)^{-2}A\Sigma A^\top)^{-1} \det(A)^{-2}AX_2A^\top) \\
&= \frac{1}{2}t_1 t_2 + \frac{1}{2}\tr(\Sigma^{-1} X_1 \Sigma^{-1} X_2)
=\g_{M,(\alpha,\Sigma)}((t_1,X_1),(t_2,X_2))
\end{align*}
Hence the action of $A\in\GL(n,\RR)$ is isometric.
\end{proof}

Now define a map
\begin{equation}
\Psi:\Pos(n,\RR)\to \RR\times\Pos_1(n,\RR),
\quad
\Sigma\mapsto (\log(\det(\Sigma)),\ \det(\Sigma)^{-1}\Sigma).
\label{eq:symmetric_isometry}
\end{equation}
Note that for $A\in\GL(n,\RR)$,
\begin{align*}
\Psi(A\cdot\Sigma) &= (\log(\det(A\Sigma A^\top)),\ \det(A\Sigma A^\top)^{-1}A\Sigma A^\top) \\
&= (\log(\det(\Sigma))+2\log(\det(A)),\ \det(A)^{-2} \det(\Sigma^{-1}) A\Sigma A^\top) \\
&= A\cdot(\det(\Sigma),\ \det(\Sigma)^{-1}\Sigma) \\
&=A\cdot\Psi(\Sigma).
\end{align*}
So the map $\Psi$ is $\GL(n,\RR)$-equivariant.

We equip the manifold $\Pos(n,\RR)$ with the restriction of the
Fisher metric \eqref{eq:Fisher_metric} of $\doN$ to $\doN(\cdot,0)$,
which is the Fisher metric $\g$ of $\doN(\cdot,0)$ by Proposition
\ref{prop:Nmu0_tot_geod}. Then $\GL(n,\RR)$ acts isometrically on $\Pos(n,\RR)$
by Lemma \ref{lem:Aff_isometry}.

\begin{prop}\label{prop:symmetric_product}
The Riemannian manifold $(\Pos(n,\RR),\g)$
is isometric to the product $(\RR\times\Pos_1(n,\RR),\g_1\times\g_2)$ of the irreducible Riemannian symmetric spaces $(\RR,\g_1)$ and $(\Pos_1(n,\RR),\g_2)$. In particular, $(\Pos(n,\RR),\g)$ is a Riemannian symmetric space.
\end{prop}
\begin{proof}
The map $\Psi$ defined in \eqref{eq:symmetric_isometry} is
the desired isometry.
In fact, $\Psi$ is $\GL(n,\RR)$-equivariant with respect to
the isometric $\GL(n,\RR)$-actions on $\Pos(n,\RR)$ and $M$,
and since $\Psi(\Sigma)=\Psi(A\cdot I)=A\cdot\Psi(I)$ (where $\Sigma=AA^\top$),
it is enough to show that $\Psi$ is an isometry at
$I\in\Pos(n,\RR)$.
So let $X,Y\in\T_I\Pos(n,\RR)\cong\Sym(n,\RR)$.
The differential of $\Psi$ at $I$ is
\begin{align*}
\d\Psi_I X &=
\left.\frac{\d}{\d t}\right|_{t=0}
(\log(\det(I+tX)),\ \det(I+tX)^{-1}(I+tX)) \\
=\ &
\left(\det(I+tX)^{-1}\frac{\d}{\d t}\det(I+tX),\right.  \\
&    -\det(I+tX)^{-2}\left.\left.\left(\left(\frac{\d}{\d t}\det(I+tX)\right)(I + tX)
        - \det(I+tX)\frac{\d}{\d t}(I+tX)\right)\right)\right|_{t=0} \\
=\ &
(\tr(X),\ X-\tr(X)I).
\end{align*}
Then
\begin{align*}
\g_{M,\Psi(I)}(\d\Psi_I X,\d\Psi_I Y)
&=
\g_{M,(0,I)}((\tr(X),X-\tr(X)I),(\tr(Y),Y-\tr(Y)I)) \\
&=
\frac{1}{2}\tr(X)\tr(Y) + \frac{1}{2}\tr((X-\tr(X)I)(Y-\tr(Y)I)) \\
&=
\frac{1}{2}\tr(X)\tr(Y) + \frac{1}{2}\tr(XY-\tr(X)Y-\tr(Y)X)+\frac{1}{2}\tr(X)\tr(Y) \\
&=
\tr(X)\tr(Y) + \frac{1}{2}\tr(XY) -\frac{1}{2}\tr(\tr(X)Y)-\frac{1}{2}\tr(\tr(Y)X) \\
&=
\frac{1}{2}\tr(XY) = \g_I(X,Y).
\end{align*}
This shows that $\Psi$ is an isometry and concludes the proof
of the proposition.
\end{proof}

\begin{cor}\label{cor:Pos_isometries}
$\Iso(\Pos(n,\RR),\g)^\circ=\GL^+(n,\RR)$.
\end{cor}
\begin{proof}
Let $G=\Iso(\Pos(n,\RR),\g)$ and let $K$ be a subgroup
of $G$ such that $G/K=\Pos(n,\RR)$.
Let $\frg$, $\frk$ denote the respective Lie algebras of
$G$, $K$, and $\sigma$ the Cartan involution.
Since $G/K$ is a symmetric product by Proposition \ref{prop:symmetric_product},
$\frg$ and $\frk$ split as a products
$\frg=\frg_1\times\frg_2$ and $\frk=\frk_1\times \frk_2$,
$\frk_i\subset\frg_i$, such that $(\frg_1,\sigma)$
and $(\frg_2,\sigma)$ are the symmetric Lie algebras
associated to $\RR$ and $\Pos(n,\RR)$, respectively
(cf.~Kobayashi \& Nomizu \cite[Section XI.5]{KN}).
Since $\Pos_1(n,\RR)=\SL(n,\RR)/\O(n)$ and $\SL(n,\RR)$ is
simple, $\frg_2=\sl(n,\RR)$ by Helgason \cite[Theorem V.4.1]{helgason}.
Hence
\[
\dim G=\dim\SL(n,\RR)+\dim\Iso(\RR,\g_\RR)=(n^2-1)+1=\dim\GL(n,\RR)
\]
and clearly $\GL(n,\RR)\subseteq G$, so that
$G^\circ=\GL(n,\RR)^\circ=\GL^+(n,\RR)$.
\end{proof}

\subsection{Bundle geometry and foliations on $\doN$}

%\begin{prop}\label{prop:cholesky}
%Let $R$ denote the group of lower triangular invertible
%$n\times n$-matrices.
%Let $\Sigma\in\Pos(n,\RR)$ and let $L\in R$ be
%such that $L L^\top=\Sigma$.
%Then the map $\Pos(n,\RR)\to R$, $\Sigma\mapsto L$, is a
%diffeomorphism.
%\end{prop}
%\begin{proof}
%The decomposition $\Sigma=LL^\top$ is the well-known Cholesky
%decomposition, and there are explicit formulas to compute the
%entries of $L$ given $\Sigma$ (e.g.~Quarteroni et al.~\cite[(3.45)]{q}).
%These formulas allow to express each entry in $L$ as a
%smooth function of $\Sigma$.
%
%Conversely, the map $L\mapsto LL^\top$ is clearly smooth,
%and by Lemma \ref{lem:Aff_isometry} the lower
%triangular matrices act simply
%transitively on $\Pos(n,\RR)$,
%hence there is a unique such $L$ for every $\Sigma$.
%\end{proof}

Let $\g$ denote the Fisher metric on $\doN$.
We can now describe the geometry of $(\doN,\g)$ in terms of
Riemannian symmetric spaces.

\begin{mthm}\label{mthm:normal_fisher}
Consider the family of $n$-variate normal distributions $\doN$ equipped
with the Fisher metric $\g$, given by \eqref{eq:Fisher_metric}.
The following hold:
\begin{enumerate}
\item
$(\doN,\g)$ is a vector bundle
\[
\RR^n\longrightarrow \doN \longrightarrow \Pos(n,\RR),
\]
where the base $\Pos(n,\RR)$ is equipped with the Fisher
metric and the fiber over $\Sigma$ is $\RR^n$
with scalar product determined by $\Sigma^{-1}$.
\item
The base $\Pos(n,\RR)$ can be identified with the
totally geodesic submanifold $\doN(\cdot,\mu_0)$ for any $\mu_0\in\RR^n$,
and it is isometric to a product of irreducible Riemannian symmetric spaces
\[
\Pos(n,\RR) = \RR\times\Pos_1(n,\RR)
\]
with the metrics on the factors given in Proposition \ref{prop:symmetric_product}.
\item
The fiber $\RR^n$ over $\Sigma_0$ can be embedded as a
parallel submanifold $\doN(\Sigma_0,\cdot)$ for any fixed
$\Sigma_0\in\Pos(n,\RR)$, and as such it is orthogonal at $(\Sigma_0,\mu_0)\in\doN$ to the embedding of the base as $\doN(\cdot,\mu_0)$.
\item
The submanifolds $\doN(\cdot,\mu)$ for all $\mu\in\RR^n$ and the
submanifolds $\doN(\Sigma,\cdot)$ for all $\Sigma\in\Pos(n,\RR)$ form
two foliations of $\doN$, the leaves of which are pointwise orthogonal
to one another.
\end{enumerate}
\end{mthm}
\begin{proof}
$\doN=\RR^n\times\Pos(n,\RR)$ is a product of differentiable manifolds,
though not of Riemannian manifolds. As such, $\doN$ is trivally a vector
bundle with base $\Pos(n,\RR)$ and fiber $\RR^n$.
By Propositions~\ref{prop:Nmu0_tot_geod} and \ref{prop:NSigma0_parallel},
the submanifold $\doN(\cdot,\mu_0)$ is
totally geodesic and the submanifold $\doN(\Sigma_0,\cdot)$ is parallel, and they are
orthogonal to each other at $(\Sigma_0,\mu_0)$.
Also, the base $\doN(\cdot,\mu_0)$ is isometric to $\Pos(n,\RR)$ for every
$\mu_0\in\RR^n$.
The metrics on base and fiber are clear from \eqref{eq:Fisher_metric}.
This proves parts (1) and (3).
Part (2) is Proposition \ref{prop:symmetric_product}.
For part (4), it is clear that $\doN$ is a union of either of these
families of submanifolds, and their pointwise orthogonality is clear
from \eqref{eq:Fisher_metric}.
\end{proof}

%%%%%%%%%
\section{The symmetric space of normal distributions}
\label{sec:symmetric}
%%%%%%%%%

Due to the difficulty of explicitely computing distances in the Fisher
metric on $\doN$, Lovri\v{c}, Min-Oo and Ruh \cite{LMR} suggested to replace
the Fisher metric of $\doN$ by the Killing metric of the symmetric space
$\Pos_1(n+1,\RR)$. For any homogeneous manifold that admits a Riemannian
metric $\kappa$ turning it into a symmetric space, we will call $\kappa$ the
\emph{Killing metric}.
Although the Fisher and the Killing metric are not isometric for $n>1$,
they are still quite similar, and distances in the Killing metric can be
computed rather easily by exploiting the geometry of the symmetric space,
as explained in \cite{LMR}.

In this section, we briefly recall the approach by Lovri\v{c} et al.~\cite{LMR} and compare the Killing metric on $\Pos_1(n+1,\RR)$ to the Fisher
metric on $\doN$.
We will find that for a lot of geodesics in the Fisher metric on $\doN$,
the geodesics in the Killing metric are good approximations at long
distances.

%%%
\subsection{On the symmetric space $\boldsymbol{\Pos_1(n+1,\RR)}$}

A diffeomorphism from $\doN$ to $\Pos_1(n+1,\RR)$ is given by
\begin{equation}
\Phi:\doN\to\Pos_1(n+1,\RR),\quad
(\Sigma,\mu)\mapsto\frac{1}{\sqrt[n+1]{\det(\Sigma)}}\begin{pmatrix}
\Sigma+\mu\mu^\top & \mu\\
\mu^\top & 1
\end{pmatrix},
\label{eq:Phi}
\end{equation}
in particular, $\dim\doN=\dim\Pos_1(n+1,\RR)$.
However, $\Phi$ is not an isometry.

The tangent space of $\doN$ at $(\Sigma,\mu)=(I_n,0)$ can be identified with
$\Sym(n,\RR)\oplus\RR^n$.
The differential of $\Phi$ at $(I_n,0)$ is given by
\begin{equation}
(X,v)\mapsto\d\Phi_{(I_n,0)}(X,v)
=
\begin{pmatrix}
X-\frac{\tr(X)}{n+1}I_n & v\\
v^\top & -\frac{\tr(X)}{n+1}
\end{pmatrix},
\label{eq:dPhi}
\end{equation}
where $X\in\Sym(n,\RR)$ and $v\in\RR^n$.

\begin{remark}
$\Phi$ is not an isometry.
For example, $\|(\lambda I_n,0)\|^{\g}_{(I_n,0)}=\frac{n}{2}\lambda^2$,
but $\|\d\Phi_{(I_n,0)}(\lambda I_n,0)\|^{\kappa}_{\Phi(I_n,0)}=\frac{1}{2}\frac{n}{n+1}\lambda^2$.
In fact, the spaces $\doN$ and $\Pos_1(n+1,\RR)$ are isometric precisely
for $n=1$.
\end{remark}

The map $\Phi$ allows us to identify the spaces $\doN$ and
$\Pos_1(n+1,\RR)$, and use the global coordinate system
$(\Sigma,\mu)$ to describe the elements of $\Pos_1(n+1,\RR)$ as well.
In the following we will do so, while suppressing the dependence on
$\Phi$ in the notation.

\begin{remark}
Note that we use a different coordinate system to the one in \cite{LMR}.
There, instead of $(\Sigma,\mu)$, the authors use coordinates $(A,\mu)$,
where $A$ is the symmetric square root of $\Sigma$, that is
$\Sigma=AA^\top$. This explains the absence of certain scalar factors in our
formulas \eqref{eq:Phi} and \eqref{eq:dPhi}. It also affects the appearance
of the metric \eqref{eq:Killing_metric} below, where in addition we use
the different scaling factor $\frac{1}{2}$ rather than $\frac{1}{4}$ for
the trace to obtain a symmetric metric from the Killing form of
$\sl(n+1,\RR)$.
\end{remark}

For $n\geq 1$, the isometry group of $\Pos_1(n+1,\RR)$ is $\SL(n+1,\RR)$,
which acts on $P\in\Pos_1(n+1,\RR)$ by
\begin{equation}
P\mapsto SPS^\top.
\label{eq:symmetric_action}
\end{equation}
The affine group $\Aff^+(n,\RR)$ also acts isometrically on $\Pos_1(n+1,\RR)$
via the homomorphic embedding
%\RQ{We need $\Aff^+$ to have a well defined $n+1$-th root of $\det(A)$.}
\begin{equation}
\Aff^+(n,\RR)\hookrightarrow\SL(n+1,\RR),\quad
(A,b)\mapsto\frac{1}{\sqrt[n+1]{\det(A)}}\begin{pmatrix}
A & b\\
0 & 1
\end{pmatrix}.
\label{eq:Aff_embedding}
\end{equation}

\begin{lem}\label{lem:Phi_equivariant}
The diffeomorphism $\Phi$ is equivariant for the $\Aff^+(n,\RR)$-actions on
$\doN$ and $\Pos_1(n+1,\RR)$.
In particular, $\Aff^+(n,\RR)$ acts transitively on $\Pos_1(n+1,\RR)$.
\end{lem}
\begin{proof}
For any $(A,b)\in\Aff^+(n,\RR)$ and $(\Sigma,\mu)\in\doN$, with \eqref{eq:action},
\begin{align*}
&\Phi((A,b)\cdot(\Sigma,\mu)) = \Phi(A\Sigma A^\top,A\mu+b) \\
&=\frac{1}{\sqrt[n+1]{\det(A)^2\det(\Sigma)}}\begin{pmatrix}
A\Sigma A^\top+(A\mu+b)(A\mu+b)^\top & A\mu+b\\
\mu^\top A^\top+b^\top & 1
\end{pmatrix}\\
&=\frac{1}{\sqrt[n+1]{\det(A)^2\det(\Sigma)}}
\begin{pmatrix}
A\Sigma A^\top + A\mu\mu^\top A^\top+b\mu^\top A^\top+A\mu b^\top+bb^\top & A\mu+b\\
\mu^\top A^\top+b^\top & 1
\end{pmatrix},
\end{align*}
and with \eqref{eq:symmetric_action} and \eqref{eq:Aff_embedding},
\begin{align*}
&(A,b)\cdot\Phi(\Sigma,\mu) =
\frac{1}{\sqrt[n+1]{\det(A)^2\det(\Sigma)}}
\begin{pmatrix}
A & b\\
0 & 1
\end{pmatrix}
\begin{pmatrix}
\Sigma+\mu\mu^\top & \mu\\
\mu^\top & 1
\end{pmatrix}
\begin{pmatrix}
A^\top & 0\\
b^\top & 1
\end{pmatrix}\\
&=\frac{1}{\sqrt[n+1]{\det(A)^2\det(\Sigma)}}
\begin{pmatrix}
A\Sigma+A\mu\mu^\top+b\mu^\top & A\mu+b\\
\mu^\top & 1
\end{pmatrix}
\begin{pmatrix}
A^\top & 0\\
b^\top & 1
\end{pmatrix}\\
&=\frac{1}{\sqrt[n+1]{\det(A)^2\det(\Sigma)}}
\begin{pmatrix}
A\Sigma A^\top+A\mu\mu^\top A^\top+b\mu^\top A^\top + A\mu b^\top + bb^\top & A\mu+b\\
\mu^\top A^\top + b^\top & 1
\end{pmatrix}.
\end{align*}
Hence $\Phi$ is $\Aff^+(n,\RR)$-equivariant.
Since $\Aff^+(n,\RR)$ acts transitively on $\doN$, it acts transitively
on $\Pos_1(n+1,\RR)$ as well.
\end{proof}

The symmetric space $\Pos_1(n+1,\RR)$ is irreducible, which means that its
Killing metric is, up to a positive multiple, determined by the Killing form
of the Lie algebra $\sl(n+1,\RR)$.
The diffeomorphism $\Phi$ allows us to identify the Killing metric
$\kappa$ with its pullback to $\doN$, and thus express it in the
$(\Sigma,\mu)$-coordinates of $\doN$.
We can choose $\kappa$ suitably scaled
such that in the $(\Sigma,\mu)$-coordinates on $\doN$, it is given at
$(\Sigma,\mu)=(I_n,0)$ by
\begin{equation}
\begin{split}
\kappa_{(I_n,0)}\bigl((X,v),(Y,w)\bigr)
&=
\frac{1}{2}\tr(\d\Phi_{(I_n,0)}(X,v)\d\Phi_{(I_n,0)}(Y,w)) \\
&=
v^\top w + \frac{1}{2}\tr(XY) - \frac{1}{2(n+1)}\tr(X)\tr(Y).
\end{split}
\label{eq:Killing_metric0}
\end{equation}
Here we used \eqref{eq:dPhi} for the differentials.
Then at any point $(\Sigma,\mu)\in\doN$, the Killing metric is given
by transporting \eqref{eq:Killing_metric0} by the action of the
affine group. We obtain
\begin{equation}
\begin{split}
&\kappa_{(\Sigma,\mu)}\bigl((X,v),(Y,w)\bigr)\\
&= v^\top\Sigma^{-1}w+\frac{1}{2}\tr(\Sigma^{-1}X\Sigma^{-1}Y)-\frac{1}{2(n+1)}\tr(\Sigma^{-1}X)\tr(\Sigma^{-1}Y).
\label{eq:Killing_metric}
\end{split}
\end{equation}
Note that we use a scaling of the Killing metric $\kappa$ different from the
one in \cite{LMR}, to make it resemble the Fisher metric on $\doN$ more
closely.
Namely, up to the term $-\frac{1}{2(n+1)}\tr(\Sigma^{-1}X)\tr(\Sigma^{-1}Y)$,
\eqref{eq:Killing_metric} resembles the Fisher metric \eqref{eq:Fisher_metric}
on $\doN$.
The similarity becomes more apparent in the following paragraph.

%%%
\subsection{Killing geodesics and Fisher geodesics in $\boldsymbol{\Pos(n,\RR)}$}

We will simply speak of \emph{Fisher geodesics} and \emph{Killing geodesics}
when referring to geodesics of the Fisher metric $\g$ and the Killing metric
$\kappa$, respectively.
Even though $(\doN,\g)$ and $(\Pos_1(n+1,\RR),\kappa)$ are not isometric,
we will see that the corresponding embeddings of the symmetric cone
$\Pos(n,\RR)$ in both spaces are affinely equivalent.

For any fixed $\mu_0\in\RR^n$, define the submanifold
\[
P_n(\mu_0)=
\Bigl\{\frac{1}{\sqrt[n+1]{\det(\Sigma)}}\begin{pmatrix}
\Sigma+\mu_0\mu_0^\top & \mu_0\\
\mu_0^\top & 1
\end{pmatrix}\,\Bigl|\,\Sigma\in\Pos(n,\RR) \Bigr\}
\]
of the symmetric space $\Pos_1(n+1,\RR)$ with
Killing metric \eqref{eq:Killing_metric}.
Clearly, $P_n(\mu_0)$, just like $\doN(\cdot,\mu_0)$, is diffeomorphic
to $\Pos(n,\RR)$.

\begin{prop}\label{prop:tot_geod2}
Consider $P_n(\mu_0)$ for any fixed $\mu_0\in\RR^n$.
\begin{enumerate}
\item
The affine transformation $(I,\mu_0)$ maps $P_n(0)$ isometrically to
$P_n(\mu_0)$. In particular, the $P_n(\mu_0)$ are isometric to each
other for all $\mu_0$.
\item
$\Phi(\doN(\cdot,\mu_0))=P_n(\mu_0)$.
\item
$P_n(\mu_0)$ is a totally geodesic submanifold of $\Pos_1(n+1,\RR)$.
\item
Let $X,Y$ be coordinate vector fields in the $\Sigma$-coordinates on
$\Pos_1(n+1,\RR)$. Then their covariant derivative
with respect to the Killing metric at the point $(\Sigma,\mu_0)$ is
\begin{equation}
\nablak_X Y = -\frac{1}{2}(X\Sigma^{-1}Y + Y\Sigma^{-1}X).
\label{eq:Killing_LC}
\end{equation}
\item
$\Phi|_{\doN(\cdot,\mu_0)}:\doN(\cdot,\mu_0)\to P_n(\mu_0)$
is an affine equivalence.
\end{enumerate}
\end{prop}

In the proof of this proposition, we use the following formulas by Skovgaard \cite[Lemma 2.3 and its proof]{skovgaard0}.
Let $X,Y,Z\in\Sym(n,\RR)$ and $\Sigma\in\Pos(n,\RR)$, and let $\partial_X$
denote the directional derivative in the direction of $X$.
Then
\begin{equation}
\begin{split}
\partial_X\tr(Y\Sigma^{-1}) &= -\tr(Y\Sigma^{-1}X\Sigma^{-1}), \\
\partial_X\tr(Y\Sigma^{-1}Z\Sigma^{-1})
&=-\bigl(\tr(Y\Sigma^{-1}X\Sigma^{-1}Z\Sigma^{-1})+\tr(Y\Sigma^{-1}Z\Sigma^{-1}X\Sigma^{-1})\bigr).
\end{split}
\label{eq:skovgaard2.3}
\end{equation}

\begin{proof}[Proof of Proposition \ref{prop:tot_geod2}]
Part (1) is straightforward to verify using \eqref{eq:symmetric_action},
\eqref{eq:Aff_embedding} and the definition of $P_n(\mu_0)$.
Part (2) is straightforward from \eqref{eq:Phi}.
If we use the relations for the Levi-Civita connection of $\kappa$ given in
\cite[(3.8)]{LMR}, the computation for part (3) is identical to
the proof of Proposition \ref{prop:Nmu0_tot_geod}.

For part (4),
Let $X,Y,Z$ be coordinate vector fields in the $\Sigma$-coordinates.
We interpret them as tangent vector fields of the
totally geodesic submanifold $P_n(\mu_0)$.
Define a covariant derivative $\tilde{\nabla}_XY$ on $P_n(\mu_0)$
by \eqref{eq:Killing_LC}.
Use \eqref{eq:skovgaard2.3} together with \eqref{eq:Killing_metric} to find
\begin{align*}
X\kappa(Y,Z) =\ & \frac{1}{2}\partial_X\tr(\Sigma^{-1}Y\Sigma^{-1}Z) \\
&-\frac{1}{2(n+1)}\bigl((\partial_X\tr(\Sigma^{-1}Y))\tr(\Sigma^{-1}Z)
+\tr(\Sigma^{-1}Y)(\partial_X\tr(\Sigma^{-1} Z))\bigr) \\
=\ & -\frac{1}{2}\bigl(\tr(\Sigma^{-1}Y\Sigma^{-1}X\Sigma^{-1}Z)+\tr(\Sigma^{-1}Y\Sigma^{-1}Z\Sigma^{-1}X)\bigr) \\
&+\frac{1}{2(n+1)}\bigl(\tr(\Sigma^{-1}Y\Sigma^{-1}X)\tr(\Sigma^{-1}Z)+\tr(\Sigma^{-1}Y)\tr(\Sigma^{-1}Z\Sigma^{-1}X)\bigr) \\
\end{align*}
and
\begin{align*}
\kappa(\tilde{\nabla}_X Y,Z) =\ & -\frac{1}{2}\kappa(X\Sigma^{-1}Y+Y\Sigma^{-1}X,Z) \\
=\ &-\frac{1}{4}\tr(\Sigma^{-1}X\Sigma^{-1}Y\Sigma^{-1}Z+\Sigma^{-1}Y\Sigma^{-1}X\Sigma^{-1}Z) \\
&+\frac{1}{4(n+1)}\bigl(\tr(\Sigma^{-1}X\Sigma^{-1}Y)\tr(\Sigma^{-1}Z)
+\tr(\Sigma^{-1}Y\Sigma^{-1}X)\tr(\Sigma^{-1}Z)\bigr), \\
=\ &-\frac{1}{4}\bigl(\tr(\Sigma^{-1}X\Sigma^{-1}Y\Sigma^{-1}Z)+\tr(\Sigma^{-1}Y\Sigma^{-1}X\Sigma^{-1}Z)\bigr) \\
&+\frac{1}{2(n+1)}\tr(\Sigma^{-1}X\Sigma^{-1}Y)\tr(\Sigma^{-1}Z), \\
\kappa(Y,\tilde{\nabla}_X Z)
=\ &-\frac{1}{4}\bigl(\tr(\Sigma^{-1}Y\Sigma^{-1}X\Sigma^{-1}Z)+\tr(\Sigma^{-1}Y\Sigma^{-1}Z\Sigma^{-1}X)\bigr) \\
&+\frac{1}{2(n+1)}\tr(\Sigma^{-1}Y)\tr(\Sigma^{-1}X\Sigma^{-1}Z).
\end{align*}
After applying some identities for the trace and collecting terms, we
find that indeed
\[
X\kappa(Y,Z)=\kappa(\tilde{\nabla}_X Y,Z)+\kappa(Y,\tilde{\nabla}_X Z).
\]
By evaluating $\tilde{\nabla}_X Y$ on coordinate vector fields, we readily
find that the torsion vanishes. Hence $\tilde{\nabla}$ is the Levi-Civita
connection of the restriction of $\kappa$ to $P_n(\mu_0)$, and since
$P_n(\mu_0)$ is a totally geodesic submanifold, $\tilde{\nabla}$ is the
restriction of the Levi-Civita connection $\nablak$ of
$(\Pos_1(n+1,\RR),\kappa)$ to $P_n(\mu_0)$.

For part (5), it is evident from comparing \eqref{eq:LeviCivita} and
\eqref{eq:Killing_LC} that $\Phi|_{\doN(\cdot,\mu_0)}$ is indeed an affine
equivalence from $\doN(\cdot,\mu_0)$ to $P_n(\mu_0)$.
%For part (5), consider first the case $\mu_0=0$.
%Since the tangent vector fields to $\doN(\cdot,0)$ are precisely
%the vector fields in the $\Sigma$-coordinates, part (4) applies to them.
%Then it is evident from comparing \eqref{eq:LeviCivita} and
%\eqref{eq:Killing_LC} that $\Phi|_{\doN(\cdot,0)}$ is indeed an affine
%equivalence from $\doN(\cdot,0)$ to $P_n(0)$.
%For arbitrary $\mu_0$ and for $(A,\mu_0)\in\Aff(n,\RR)$, consider the
%following diagram:
%\[
%\xymatrix{
%\doN(\cdot,0) \ar[d]_\Phi \ar[r]^{(A,\mu_0)} & \doN(\cdot,\mu_0) \ar[d]^\Phi \\
%P_n(0) \ar[r]^{(A,\mu_0)} & P_n(\mu_0)
%}
%\]
%This diagram commutes by Lemma \ref{lem:Phi_equivariant}.
%The horizontal arrows are isometries and thus map geodesics to
%geodesics, and the left vertical arrow was just seen to be an affine
%equivalence.
%Hence the right vertical arrow is also an affine equivalence.
\end{proof}

%For part (4), a geodesic $\gamma$ in $\doN(\cdot,0)$ through $(I_n,0)$
%is the orbit $(\Sigma(t),0)$ of a one-parameter group
%$\Sigma(t)=\exp(tX)$ for some $X\in\Sym(n,\RR)$. The map $\Phi$ maps
%the orbit $(\Sigma(t),0)$ through $(I_n,0)$ to the orbit of $\Sigma(t)$
%through $\Phi(I_n,0)=I_{n+1}$, and its differential maps $X\in\Sym(n,\RR)$
%to $\d\Phi_{(I_n,0)}X=\left(\begin{smallmatrix}X-\frac{\tr(X)}{n+1}I_n&0\\0&-\frac{\tr(X)}{n+1}\end{smallmatrix}\right)\in\Sym_0(n+1,\RR)$. The
%summands in the Cartan decomposition of $\sl(n+1,\RR)$ are $\frk=\so(n+1)$
%and $\frm=\Sym_0(n+1,\RR)$, so $\Phi(\Sigma(t),0)$ is indeed the orbit
%of a one-parameter group tangent to $\frm$ at the identity.
%Hence $\Phi\circ\gamma$ is a Killing geodesic in $P_n(0)$ through $I_{n+1}$.
%Now consider the commutative (by Lemma \ref{lem:Phi_equivariant}) diagram
%\[
%\xymatrix{
%\doN(\cdot,0) \ar[d]_\Phi \ar[r]^{(A,\mu_0)} & \doN(\cdot,\mu_0) \ar[d]^\Phi \\
%P_n(0) \ar[r]^{(A,\mu_0)} & P_n(\mu_0)
%}
%\]
%for $(A,\mu_0)\in\Aff(n,\RR)$.
%Here, the horizontal arrows are isometries and thus map geodesics to
%geodesics, and the left vertical arrow was just seen to map Fisher geodesics
%at $(I_n,0)$ to Killing geodesics at $I_{n+1}$.
%Hence the right vertical arrow also maps geodesics to geodesics.
%By the transitivity of the $\Aff(n,\RR)$-action, it follows that all
%Fisher geodesics in any $\doN(\cdot,\mu_0)$ get mapped to Killing geodesics
%in the corresponding $P_n(\mu_0)$.

%%%
\subsection{Distances in $\boldsymbol{\Pos(n,\RR)}$}

Distances between points contained in $\doN(\cdot,\mu_0)$ for fixed
$\mu_0\in\RR$ are readily computed using the fact that $\doN(\cdot,\mu_0)$
is a totally geodesic submanifold of $\doN$, and also a Riemannian symmetric
space isometric to $\Pos(n,\RR)$.
The distances in this symmetric space are easy to compute, since they
can be reduced to computations in a flat totally geodesic submanifold.

\begin{lem}\label{lem:dist_diag}
Let $\Delta=\diag(\delta_1,\ldots,\delta_n)\in\Diag(n,\RR)\cap\Pos(n,\RR)$.
The Fisher distance from the identity matrix $I_n$ to $\Delta$ is
\[
\dist_{\g}(I_n,\Delta) = \sqrt{\frac{1}{2}\sum_{i=1}^n\log(\delta_i)^2}.
\]
\end{lem}
\begin{proof}
It is well-known that $\Diag(n,\RR)\cap\Pos(n,\RR)$ is a maximal flat
totally geodesic subspace in $\Pos(n,\RR)$. Thus the geodesic $\gamma$
from $I_n$ to $\Delta$ is
\[
\gamma(t) = \exp(t\Lambda),
\]
where $\Lambda=\log(\Delta)$ (this is well-defined since all eigenvalues
of $\Delta$ are positive).
Then by \eqref{eq:Fisher_metric} for all $t$,
%\RQ{Using \eqref{eq:Fisher_metric},
%I think the following should be $\sqrt{\frac{\tr(\Lambda^2)}{2}}$. This makes our formula
%different from \cite{LMR} by a factor of $\sqrt{\frac{1}{2}}$, which seems to be right since our metrics on the cone differ by a factor of $\frac{1}{2}$.}
\[
\|\gamma'(t)\|_{\gamma(t)}^{\g} = \|\Lambda\gamma(t)\|_{\gamma(t)}^{\g} =
\sqrt{\frac{\tr(\Lambda^2)}{2}}.
\]
The distance from $I_n$ to $\Delta$ is then
\[
\dist_{\g}(I_n,\Delta)=\int_0^1\|\gamma'(t)\|^{\g}_{\gamma(t)}\d t
=\sqrt{\frac{\tr(\Lambda^2)}{2}}
=\sqrt{\frac{1}{2}\sum_{i=1}^n\log(\delta_i)^2}
\]
with $\Lambda=\diag(\log(\delta_1),\ldots,\log(\delta_n))$.
\end{proof}

Similar to the procedure described by Lovri\v{c} et al.~\cite[pp.~42-43]{LMR}
for $\Pos_1(n+1,\RR)$ with the Killing metric,
we describe the procedure to derive the Fisher distance formula for
elements $S_1=(\Sigma_1,\mu_0)$, $S_2=(\Sigma_2,\mu_0)$ in
$\doN(\cdot,\mu_0)$,
\begin{enumerate}
\item
By applying the isometry $(I_n,-\mu_0)\in\Aff^+(n,\RR)$, we may assume
that $S_1,S_2\in\doN(\cdot,0)$. Under the identification of this submanifold
with $\Pos(n,\RR)$, we identify $S_i$ with $\Sigma_i$.
\item
We can write $\Sigma_1=A_1 A_1^\top$ for some $A_1\in\SL(n,\RR)$.
\item
When applying the isometry $(A_1^{-1},0)\in\Aff^+(n,\RR)$, we have
\[
\dist_{\g}(\Sigma_1,\Sigma_2)=\dist_{\g}(I_n,A_1^{-1}\Sigma_2 A_1^{-\top}).
\]
We may thus assume that $\Sigma_1=I_n$. Note that this may change the
eigenvalues of $\Sigma_2$.
\item
By applying an isometry $(T,0)$ for some $T\in\O(n)$, we may assume that
$A_1^{-1}\Sigma_2 A_1^{-\top}=\Delta$ is a diagonal matrix in
$\Diag(n,\RR)\cap\Pos(n,\RR)$.
Now Lemma \ref{lem:dist_diag} applies, and we obtain
\begin{equation}
\dist_{\g}(S_1,S_2)
=\dist_{\g}(\Sigma_1,\Sigma_2)
=\sqrt{\frac{1}{2}\sum_{i=1}^n\log(\lambda_i)^2}
\end{equation}
for the eigenvalues $\lambda_1,\ldots,\lambda_n$ of the matrix
$A_1^{-1}\Sigma_2^2 A_1^{-\top}$.
\end{enumerate}

Up to a factor $\frac{1}{\sqrt{2}}$, this coincides with the distance
formula for elements in $P_n(\mu_0)$ as computed in \cite{LMR}.

%%%
\subsection{Asymptotic geodesics orthogonal to $\boldsymbol{\Pos(n,\RR)}$}

As we just saw, the lengths of geodesics in $\doN$ tangent to the symmetric
submanifolds $\doN(\cdot,\mu_0)$ are rela\-tively easy to compute.
Unfortunately, the same cannot be said for geodesics transversal to
$\doN(\cdot,\mu_0)$. Although explicit solutions for the Fisher metric's
geodesic equation have been found by Calvo and Oller \cite[Section 3]{CO2},
they only yield explicit formulas for the distance between two points in
some special cases. In this paragraph, we want to argue that Killing
geodesics provide reasonable approximations whose lengths are easy to
compute.

We introduce some terminology. Let $c:\RR\to\doN$ be a differentiable
curve and define the \emph{geodesic defect} of $c$ to be
\[
\delta(c) = \lim_{t\to\infty}\frac{1}{t}\int_0^t\|\nabla_{c'(s)}c'(s)\|^{\g}_{c(s)}\d s.
\]
If $\delta(c)=0$, then we call $c$ an \emph{asymptotic geodesic} in the
Fisher metric on $\doN$. Note that by this definition, $\delta(c)$ is
invariant under isometries of $\doN$.
We restrict ourselves to curves with domain of definition $\RR$ here, since
below we will only study Killing geodesics $c$, which are complete.

Our goal in this paragraph is to compare the behaviour of such Killing
geodesics with that of Fisher geodesics, and eventually we will show:

\begin{mthm}\label{mthm:asymptotic}
Consider the family of $n$-variate normal distributions $\doN$ equipped
with the Fisher metric $\g$, given by \eqref{eq:Fisher_metric}.
Let $c:\RR\to\doN$ be a geodesic for the Killing metric $\kappa$ on $\doN$,
given by \eqref{eq:Killing_metric}.
Assume that $c(0)=(\Sigma_0,\mu_0)$ and $c'(0)\perp\doN(\cdot,\mu_0)$.
Then $c$ is an asymptotic geodesic for the Fisher metric.
\end{mthm}

The proof requires some preparations.
For simplicity, we will assume that
\[
c(0)=(I_n,0),
\quad
c'(0)=(0,e_1).
\]
In the proof of Theorem \ref{mthm:asymptotic} below we see that it is
sufficient to treat this case.

At the point $(I_n,0)$, the tangent subspace orthogonal to
$\T_{(I_n,0)}\doN(\cdot,0)$ is mapped by $\d\Phi$ to
\[
V=
\Bigl\{
\begin{pmatrix}
0 & v\\
v^\top & 0
\end{pmatrix}\ \Bigl|\ v\in\RR^n
\Bigr\}.
\]
Incidentally, $V$ is also the orthogonal space to $\T_{I_{n+1}}P_n(0)$
for the Killing metric on $\Pos_1(n+1,\RR)$.
Moreover, $V$ lies in $\Sym_0(n+1,\RR)$, the complement of the maximal
subalgebra of compact type in the Cartan decomposition of $\sl(n+1,\RR)$.

\begin{remark}
Recall that in any Rie\-mannian symmetric space $M=G/K$, the geodesics
through a point $p\in M$ are given as the orbits of one-parameter subgroups
\begin{equation}
\gamma(t)=\exp(tX)p
\quad
\text{ for some } X\in\frm,
\label{eq:geodesic_opg}
\end{equation}
where $\frg=\frk\oplus\frm$ is a Cartan decomposition
of the Lie algebra of $G$ (cf.~Kobayashi \& Nomizu \cite[Corollary X.2.5]{KN}).
In particular, for $M=\Pos(n,\RR)$ and $G=\GL(n,\RR)$, the subspace
$\frm$ is $\Sym(n,\RR)$, and for $M=\Pos_1(n,\RR)$ and $G=\SL(n,\RR)$,
the subspace $\frm$ is $\Sym_0(n,\RR)$.
\end{remark}

By this remark, the Killing geodesics tangent to $V$ at $I_{n+1}$ in
$\Pos_1(n+1,\RR)$ is given by the action \eqref{eq:symmetric_action}
of the one-parameter subgroups
\[
\exp\begin{pmatrix}
0 & tv\\
tv^\top & 0
\end{pmatrix}.
\]

\begin{lem}\label{lem:exponential_geodesic}
The Killing geodesic $\tilde{c}$ with
\[
\tilde{c}(0)=I_{n+1},
\quad
\tilde{c}'(0)=
\begin{pmatrix}
0 & e_1\\
e_1^\top & 0
\end{pmatrix}
\]
is given by
\begin{equation}
\tilde{c}(t)
=
\begin{pmatrix}
\cosh(2t) & 0 & \sinh(2t) \\
0 & I_{n-1} & 0 \\
\sinh(2t) & 0 & \cosh(2t)
\end{pmatrix}.
\label{eq:exponential_geodesic}
\end{equation}
Its preimage in $\doN$ under the diffeomorphism $\Phi$ is
\begin{equation}
c(t)=
(\Phi^{-1}\circ\tilde{c})(t)
=\Bigl(\begin{pmatrix}
\cosh(2t)^{-2} & 0\\
0 & \cosh(2t)^{-1}I_{n-1}
\end{pmatrix},\tanh(2t)e_1\Bigr).
\label{eq:exponential_geodesic2}
\end{equation}
\end{lem}
\begin{proof}
Write $X=\tilde{c}'(0)$.
By induction, we find that the even and odd powers of $X$ are
%\RQ{The case $k=0$ is special for even powers.}
\[
X^{2k} = \begin{pmatrix}
e_1 e_1^\top & 0\\
0 & 1
\end{pmatrix}, k \geq 1, \quad
X^{2k+1}= \begin{pmatrix}
0 & e_1\\
e_1^\top & 0
\end{pmatrix}, k \geq 0.
\]
Since $e_1 e_1^\top = E_{11}$ we have
\begin{align*}
\exp(tX) &= \sum_{k=0}^\infty\frac{t^{2k+1}}{(2k+1)!}X^{2k+1}
+\sum_{k=0}^\infty\frac{t^{2k}}{(2k)!}X^{2k} \\
&= \begin{pmatrix}
0 & 0 & \sinh(t)\\
0 & 0 & 0 \\
\sinh(t) & 0 & 0
\end{pmatrix}
+\begin{pmatrix}
\cosh(t)& 0 & 0 \\
0 & I_{n-1} & 0 \\
0 & 0 & \cosh(t)
\end{pmatrix}  \\
&=
\begin{pmatrix}
\cosh(t) & 0 & \sinh(t) \\
0 & I_{n-1} & 0 \\
\sinh(t) & 0 & \cosh(t)
\end{pmatrix}.
\end{align*}
This one-parameter subgroup acts on $I_{n+1}$ by
\[
\exp(tX)I_{n+1}\exp(tX)^\top = \exp(tX)^2 = \exp(2tX)
=
\begin{pmatrix}
\cosh(2t) & 0 & \sinh(2t) \\
0 & I_{n-1} & 0 \\
\sinh(2t) & 0 & \cosh(2t)
\end{pmatrix}.
\]
which is the desired expression \eqref{eq:exponential_geodesic} for the geodesic $\tilde{c}$.
To obtain the expression for $c$, we need the $(\Sigma,\mu)$-coordinates
of $\tilde{c}$. By \eqref{eq:Phi},
\[
\frac{1}{\sqrt[n+1]{\det(\Sigma)}}=\cosh(2t),
\quad
\mu=\tanh(2t)e_1,
\]
and thus
\[
\Sigma=\begin{pmatrix}
1 &0\\
0& \cosh(2t)^{-1}I_{n-1}
\end{pmatrix}
-\tanh(2t)^2 e_1e_1^\top
=
\begin{pmatrix}
\cosh(2t)^{-2}&0\\
0&\cosh(2t)^{-1}I_{n-1}
\end{pmatrix}.
\]
This yields the expression \eqref{eq:exponential_geodesic2} for $c(t)$.
\end{proof}

After applying some identities for the hyperbolic functions, we find:

%\RQ{I made some slight changes in the formulas below, so that it is easier to see how to get Lemma \ref{lem:covar_c}}

\begin{lem}\label{lem:nasty_derivatives}
The first and second derivatives of the Killing geodesic $c$ are
\begin{align}
c'(t)
&=\Bigl(\begin{pmatrix}
-\frac{4\sinh(2t)}{\cosh(2t)^3} & 0\\
0 & -\frac{2\sinh(2t)}{\cosh(2t)^2}I_{n-1}
\end{pmatrix},\frac{2}{\cosh(2t)^2}e_1\Bigr),
\label{eq:dcdt} \\
c''(t)
&=\Bigl(\begin{pmatrix}
\frac{-8+16\sinh(2t)^2}{\cosh(2t)^4} & 0\\
0 & \frac{-4+4\sinh(2t)^2}{\cosh(2t)^3}I_{n-1}
\end{pmatrix},-\frac{8\sinh(2t)}{\cosh(2t)^3}e_1\Bigr).
\label{eq:d2cdt2}
\end{align}
\end{lem}

Using Lemma \ref{lem:nasty_derivatives} and \eqref{eq:LeviCivita}, we can
now compute the second covariant derivative of $c(t)=(\Sigma(t),\mu(t))$,
\begin{equation}
\nabla_{c'(t)}{c'(t)}
=c''(t)-\bigl(\Sigma'(t)\Sigma^{-1}(t)\Sigma'(t)-\mu'(t)\mu'(t)^\top,\ \Sigma'(t)\Sigma^{-1}(t)\mu'(t)\bigr),
\label{eq:covar_c_0}
\end{equation}
with
\begin{align*}
&\Sigma'(t)\Sigma^{-1}(t)\Sigma'(t)\\
&=\begin{pmatrix}
-\frac{4\sinh(2t)}{\cosh(2t)^3} & 0\\
0 & -\frac{2\sinh(2t)}{\cosh(2t)^2}I_{n-1}
\end{pmatrix}
\begin{pmatrix}
\cosh(2t)^{2} & 0\\
0 & \cosh(2t)I_{n-1}
\end{pmatrix}
\begin{pmatrix}
-\frac{4\sinh(2t)}{\cosh(2t)^3} & 0\\
0 & -\frac{2\sinh(2t)}{\cosh(2t)^2}I_{n-1}
\end{pmatrix}\\
&=
\begin{pmatrix}
\frac{16\sinh(2t)^2}{\cosh(2t)^4}&0\\
0&\frac{4\sinh(2t)^2}{\cosh(2t)^3}I_{n-1}
\end{pmatrix},\\
&\mu'(t)\mu'(t)^\top =
\frac{4}{\cosh(2t)^4}e_1e_1^\top
=
\begin{pmatrix}
\frac{4}{\cosh(2t)^4}&0\\
0&0
\end{pmatrix}, \\
&\Sigma'(t)\Sigma^{-1}(t)\mu'(t) \\
&=\frac{2}{\cosh(2t)^2}\begin{pmatrix}
-\frac{4\sinh(2t)}{\cosh(2t)^3} & 0\\
0 & -\frac{2\sinh(2t)}{\cosh(2t)^2}I_{n-1}
\end{pmatrix}
\begin{pmatrix}
\cosh(2t)^{2} & 0\\
0 & \cosh(2t)I_{n-1}
\end{pmatrix}
e_1 \\
&=
-\frac{8\sinh(2t)}{\cosh(2t)^3}e_1.
\end{align*}

We substitute these expressions and \eqref{eq:d2cdt2} in
\eqref{eq:covar_c_0} to obtain:

\begin{lem}\label{lem:covar_c}
The second covariant derivative $\nabla_{c'(t)}c'(t)$ for the Fisher metric
in $\doN$ of the Killing geodesic $c$ is
\begin{equation}
\nabla_{c'(t)}c'(t)
=
\Bigl(
\begin{pmatrix}
-\frac{4}{\cosh(2t)^4} & 0 \\
0 & -\frac{4}{\cosh(2t)^3} I_{n-1}
\end{pmatrix}
, 0\Bigr).
\label{eq:covar_c}
\end{equation}
In particular, $c$ is not a Fisher geodesic.
\end{lem}

With this lemma, we can prove Theorem \ref{mthm:asymptotic}.

\begin{proof}[Proof of Theorem \ref{mthm:asymptotic}]
Let $c$ be a Killing geodesic beginning at a point $c(0)=(\Sigma_0,\mu_0)$
whose initial direction $c'(0)$ is orthogonal to $\doN(\cdot,\mu_0)$.
That is, $c'(0)=(0,v)$ for some $v\in\RR^n$.
\begin{enumerate}
\item
We may reparameterize $c$ by rescaling the parameter $t$ such that
$\|v\|^{\g}_{(\Sigma_0,\mu_0)}=1$.
This affects the second covariant derivative of $c$ only by a constant
factor.
\item
By applying an isometry $(A,b)\in\Aff^+(n,\RR)$ with $A^\top A=\Sigma_0^{-1}$
and $b=-A\mu_0$, we may assume that $c(0)=(I_n,0)$ and
$c'(0)$ is orthogonal to $\doN(\cdot,0)$.
\item
Then we may apply another isometry $(T,0)\in\Aff(n,\RR)$ with
$T\in\O(n)$, so that we may assume $c'(0)=(0,e_1)$, while
$c(0)=(I_n,0)$ still holds.
\end{enumerate}
Since the affine group acts isometrically for both the Fisher metric and
the Killing metric, the resulting curve $c$ is still a Killing geodesic.

By Lemma \ref{lem:covar_c} and \eqref{eq:Fisher_metric},
\begin{align*}
\|\nabla_{c'(t)}c'(t)\|^{\g}_{c(t)}
&=\sqrt{\frac{1}{2}\tr\bigl(c(t)^{-1} (\nabla_{c'(t)}c'(t)) c(t)^{-1} (\nabla_{c'(t)}c'(t))\bigr)} \\
&=\sqrt{\frac{1}{2}\tr\begin{pmatrix}
\frac{16}{\cosh(2t)^4} & 0 \\
0 & \frac{16}{\cosh(2t)^4}I_{n-1}
\end{pmatrix}} \\
&=\frac{2\sqrt{2n}}{\cosh(2t)^2}.
\end{align*}
Now
\[
\int_0^t \|\nabla_{c'(s)}c'(s)\|^{\g}_{c(s)} \d s
=
2\sqrt{2n}\int_0^t\frac{1}{\cosh(2s)^2} \d s
=
\sqrt{2n}\tanh(2t).
\]
It follows that
\[
\delta(c)
=
\lim_{t\to\infty}\frac{1}{t}\int_0^t \|\nabla_{c'(s)}c'(s)\|^{\g}_{c(s)} \d s \\
=
\lim_{t\to\infty}\frac{\sqrt{2n}\tanh(2t)}{t} =0.
\]
Hence the Killing geodesic $c$ is an asymptotic Fisher geodesic.
As the geodesic defect is invariant under isometries of the Fisher metric,
this is true for any geodesic with $c'(0)$ orthogonal to
$\doN(\cdot,\mu_0)$.
\end{proof}

%%%%% Bibliography %%%%%%


\begin{thebibliography}{99}

%\bibitem{AJLS} N. Ay, J. Jost, H.V. L\^e, L. Schwachh\"ofer,
%\emph{Information Geometry},
%Ergebnisse der Mathematik und ihrer Grenzgebiete, Series 3, 64,
%Springer (2017)

\bibitem{AN} S. Amari, H. Nagaoka,
\emph{Methods of Information Geometry},
Translations of Mathematical Monographs 191,
American Mathematical Society (2000)

\bibitem{CO1} M. Calvo, J.M. Oller,
\emph{A Distance between Multivariate Normal Distributions Based in an Embedding into the Siegel Group},
Journal of Multivariate Analysis 35 (1990), 223-242

\bibitem{CO2} M. Calvo, J.M. Oller,
\emph{An explicit solution of information geodesic equations for the multivariate normal model},
Statistics \& Decisions 9 (1991), 119-138

\bibitem{CSS} S.I.R. Costa, S.A. Santos, J.E. Strapasson,
\emph{Fisher information distance: A geometrical rea\-ding},
Discrete Applied Mathematics 197 (2015), 59-69

\bibitem{EM} M. Eastwood, V. Matveev,
\emph{Metric connections in projective differential geometry},
in ``Symmetries and Overdetermined Systems of Partial Differential Equations'', 339-351, IMA Volumes in Mathematics and its Applications 144,
Springer (2007)

%\bibitem{fisher} R.A. Fisher,
%\emph{On the Mathematical Foundations of Theoretical Statistics},
%Philosophical Transaction of the Royal Society of London A 222 (1922), 309-368

\bibitem{helgason} S. Helgason,
\emph{Differential Geometry, Lie Groups, and Symmetric Spaces},
Graduate Studies in Mathematics 34, American Mathematical Society
(2001)

\bibitem{KN} S. Kobayashi, K. Nomizu,
\emph{Foundations of Differential Geometry I \& II},
John Wiley \& Sons (1963 \& 1969)

\bibitem{LMR} M. Lovri\v{c}, M. Min-Oo, E.A. Ruh,
\emph{Multivariate Normal Distributions Parameterized as a
Riemannian Symmetric Space},
Journal of Multivariate Analysis 74 (2000), 36-48

%\bibitem{q} A. Quarteroni, R. Sacco, F. Saleri,
%\emph{Numerical Mathematics}, second edition,
%Texts in applied mathematics, Springer (2007)

%\bibitem{rao} C.R. Rao,
%\emph{Information and the accuracy attainable in the estimation of statistical parameters},
%Bulletin of the Calcutta Mathematical Society 37 (1945), 81-91

%\bibitem{shima} H. Shima,
%\emph{The geometry of Hessian structures},
%World Scientific (2007)

\bibitem{skovgaard0} L.T. Skovgaard,
\emph{A Riemannian Geometry of the multivariate normal model},
Research Report 81/3.
Statistical Research Unit, Danish Medical Research
Council, Danish Social Science Research Council (1981)

\bibitem{skovgaard} L.T. Skovgaard,
\emph{A Riemannian Geometry of the Multivariate Normal Model},
Scandinavian Journal of Statistics 11 (1984), no.~4, 211-223

\end{thebibliography}
\end{document}